\documentclass[journal]{IEEEtran}

\usepackage{float}
\usepackage[colorlinks,linkcolor=blue]{hyperref}
\usepackage[latin1]{inputenc}
\usepackage{amsmath}
\usepackage{graphicx}
\usepackage{amssymb}
\usepackage{amsthm}
\usepackage{verbatim}
\usepackage{epsfig}
\usepackage[labelfont=bf]{caption}
\usepackage{subcaption}
\usepackage{enumerate}

\usepackage{booktabs}
\usepackage{stfloats}
\usepackage{amssymb}
\usepackage{amsthm}
\usepackage{graphicx}
\usepackage{float}
\usepackage{amsmath}
\usepackage{amsmath,bm}
\usepackage{color}
\usepackage{multirow}
\usepackage{mathrsfs}
\usepackage{lineno,hyperref}
\usepackage{lipsum}
\usepackage{stfloats}

\usepackage{dsfont}
\usepackage{mathtools}

\newtheorem{definition}{Definition}
\newtheorem{assumption}{Assumption}
\newtheorem{lemma}{Lemma}
\newtheorem{remark}{Remark}

\newtheorem{proposition}{Proposition}

\def\begquo{\begin{quote}}
	\def\endquo{\end{quote}}
\def\begequarr{\begin{eqnarray}}
\def\endequarr{\end{eqnarray}}
\def\begequarrs{\begin{eqnarray*}}
	\def\endequarrs{\end{eqnarray*}}
\def\begarr{\begin{array}}
	\def\endarr{\end{array}}
\def\begequ{\begin{equation}}
\def\endequ{\end{equation}}
\def\lab{\label}
\def\begdes{\begin{description}}
	\def\enddes{\end{description}}
\def\begenu{\begin{enumerate}}
	\def\begite{\begin{itemize}}
		\def\endite{\end{itemize}}
	\def\endenu{\end{enumerate}}

\def\lef[{\left[\begin{array}}
	\def\rig]{\end{array}\right]}

\def\begcen{\begin{center}}
	\def\endcen{\end{center}}
\def\begrem{\begin{remark}\rm}
	\def\endrem{\end{remark}}
\def\begdef{\begin{definition}}
	\def\enddef{\end{definition}}
\def\begpro{\begin{propositionosition}}
	\def\endpro{\end{propositionosition}}
\def\begfac{\begin{fact}}
	\def\endfac{\end{fact}}
\def\begass{\begin{assumptionption}}
	\def\endass{\end{assumptionption}}

\def\begmat#1{\begin{bmatrix}#1\end{bmatrix}}
\def\begali#1{\begin{align}{#1}\end{align}}
\def\begalis#1{\begin{align*}{#1}\end{align*}}

\def\liminf{\lim_{t \to \infty}}

\def\L2e{{\cal L}_{2e}}

\def\rea{\mathbb{R}}

\def\adj{\mbox{adj}}
\def\col{\mbox{col}}

\def\TAC{{\it IEEE Transactions Automatic Control}}

\def\TPS{{\it IEEE Transactions on Power Systems}}

\usepackage{color}

\usepackage[prependcaption,colorinlistoftodos]{todonotes}

\title{PMU-Based Decentralized Mixed Algebraic and Dynamic State Observation in Multi-Machine Power Systems}
\author{M. Nicolai L. Lorenz-Meyer, Alexey A. Bobtsov, Romeo Ortega, Nikolay Nikolaev, and Johannes Schiffer
	\thanks{M. N. L. Lorenz-Meyer and J. Schiffer are with Brandenburg University of Technology Cottbus-Senftenberg, 03046 Cottbus, Germany (e-mail: lorenz-meyer@b-tu.de).}
	\thanks{A. A. Bobtsov and N. Nikolaev are with ITMO University, 197101, Saint Petersburg, Russian Federation.}
	\thanks{R. Ortega is with Departamento Acad\'{e}mico de Sistemas Digitales, ITAM, Ciudad de M\'exico, M\'{e}xico.}
}

\begin{document}
	\maketitle

\begin{abstract}
We propose a novel decentralized mixed algebraic and dynamic state observation method for multi-machine power systems with unknown inputs and equipped with Phasor Measurement Units (PMUs). More specifically, we prove that for the third-order flux-decay model of a synchronous generator, the local PMU measurements give enough information to reconstruct {\em algebraically} the load angle and the quadrature-axis internal voltage. Due to the algebraic structure a high numerical efficiency is achieved, which makes the method applicable to large scale  power systems. 
Also, we prove that the relative shaft speed can be globally estimated combining a classical Immersion and Invariance (I\&I) observer with---the recently introduced---dynamic regressor and mixing (DREM) parameter estimator. This adaptive observer ensures global convergence under weak excitation assumptions that are verified in applications. The proposed method does not require the measurement of exogenous inputs signals such as the field voltage and the mechanical torque nor the knowledge of mechanical subsystem parameters. 
\end{abstract}

\section{Introduction}

\subsection{Motivation and Existing Literature} 
Situational awareness of the transient dynamics of power systems is becoming increasingly important as these systems undergo major changes due to the massive introduction of power-electronics-interfaced components on the generation side, while at the same time novel demand-response technologies and more complex loads are implemented on the customer side \cite{ZHAetal, winter15}. Furthermore, the systems operate more frequently at high loading and, thus, closer to the stability limits \cite{winter15, ulbig14}. Additionally, cascaded failures and disconnections can occur due to improper protection and control methods in such stressed conditions~\cite{winter15, ulbig14}. 
It is, hence, highly relevant for power system operators to reliably monitor the system states in real time, especially during and after disturbances \cite{ZHAetal, dehghani14}.

These developments render the steady-state assumptions used in traditional static state estimation questionable. Thus, the deployment of methods for dynamic state estimation (DSE) is gaining increasing importance for power system control and protection. This is further emphasized, as the dynamics of power systems become faster and more volatile due to the abovementioned changes~\cite{milano18}. A main application of DSE and, hence, the proposed method in the present paper lies within the field of dynamic security assessment (DSA). For a reliable DSA, computationally efficient and robust DSE techniques are essential, not only to accurately monitor the system states but also to precisely initialize transient simulations \cite{ZHAetal}.

A key enabler for DSE is the increasing availability of phasor measurement units (PMUs) \cite{ZHAetal}. Their ability to provide time-stamped, high-frequency measurements of voltages, currents and powers \cite{matavalam19,anderson14} gives rise to modern dynamic observer designs.

As of today, the prevalent algorithms in the literature for DSE are Kalman Filter-based techniques. This includes the Extended Kalman Filter (EKF), which is used in \cite{paul18} to estimate the states of a synchronous generator as well as the exciter field voltage and the output mechanical torque of the governor using PMU measurements from the bus closest to the generator. To avoid the linearization and Jacobian matrix calculation required for the EKF, several derivative free state estimation methods employing Unscented Kalman Filters (UKF) have been proposed, see \cite{wang11,valverde11}. The UKF is based on the application of the unscented transformation combined with the Kalman filter theory. The authors of \cite{anagnostou17} extended the application of this method to the case of a synchronous generator with unknown inputs.  
In \cite{Emami2015, Cui2015}, purely probability-based Particle Filters (PF) are utilized. In this way, smooth and accurate state estimation results may even be achieved in the  presence of unknown changes in the covariance of the measurement noise. For a thorough review of various Kalman Filter-based approaches to DSE, the reader is referred to \cite{ZHAetal, SINPAL}. Yet, in all the abovementioned publications, the observer convergence is only shown empirically, {\em i.e.} via simulations, but no rigorous convergence guarantees are provided. 

The problem of state observation of nonlinear systems has been extensively studied in the control systems literature, providing answers to many important theoretical questions and proposing effective observer designs with analytical convergence guarantees. We refer the reader to \cite{ASTetal,BER} for a review of the literature. 
Based on these advancements, recently efforts have been made to develop novel DSE methods based on state observer design concepts for nonlinear systems with provable convergence guarantees for the estimation errors. 
In \cite{Nugrohoetal}, a robust observer providing performance guarantees for the state estimation error norm against worst case disturbances is proposed based on the Lipschitz property of the employed synchronous generator model and the PMU measurements. However, the paper focuses only on the estimation of the generator states and considers the exogenous input signals, {\em i.e.} the field voltage and the mechanical torque, known. Using a linearized representation for the power system, in \cite{Taha18} a risk mitigation strategy to eliminate threats from the power system's unknown inputs and cyber-attacks is proposed utilizing a sliding mode observer. A nonlinear observer and a cubature KF are developed and compared to classical KF-based DSE methods in \cite{Qietal}. Through numerical simulations, the observer and the cubature KF are shown to be more robust to model uncertainties and cyber attacks against PMU measurements. The authors in \cite{Anagnostou18} introduce an anomaly detection scheme to spot inconsistent power system operational changes based on a state observer, with guaranteed estimation error convergence. However, the mechanical and electrical subsystem parameters are considered known in these proposals. In practical applications, these parameters are typically only known from baseline testing, but can significantly deviate from their nominal values due to aging, component replacements and initial measurement errors \cite{marchi2020}.  

\subsection{Contributions}
Motivated by the abovementioned recent developments, we propose a novel, unknown-input DSE method suitable for multi-machine power systems, in which the synchronous generators are represented by the three-dimensional flux-decay model \cite{kundur94,MACetal,SAUetal, Van_cutsem}. In particular, the measurement of exogenous input signals such as the field voltage and the mechanical torque, which are difficult to access for power system operators, is not required. Furthermore, in contrast to the previously mentioned methods, the only assumed prior knowledge on the generator parameters is the transient admittance magnitude. The mechanical and other electrical subsystem parameters are estimated online along with the system states. Additionally, the load angle and the quadrature-axis internal voltage are reconstructed algebraically and, thus, instantaneously and numerically highly efficient without the need of solving a differential equation. Thereby, this algebraic observer does not require initial conditions. 

By exploiting the huge potential provided by PMUs, we provide a numerically efficient method suitable for large-scale systems, while providing rigorous analytical convergence guarantees for the employed state estimation as well as the parameter identification technique. In this setting, our main contributions are three-fold: 
\begenu
\item We propose a  decentralized {\em algebraic} observer for the load angle and the quadrature-axis internal voltage using only the measurements provided via PMUs at the terminal buses. In order to establish this result, we assume that the direct-axis transient reactance and the quadrature-axis reactance are equal.
\item A certainty equivalent adaptive observer for the relative shaft speed is designed by combining the  classical Immersion and Invariance (I\&I) technique \cite{ASTetal} with the recently introduced Dynamic Regressor and Mixing (DREM) procedure \cite{ARAetal}, which has been successfully applied to several engineering problems \cite{bobtsov15,bobtsov17,SCHARIORT}. This observer is decentralized, does not require any additional prior knowledge on the system and its convergence is ensured via very weak excitation assumptions. 
\item The effectiveness of the proposed DSE technique is demonstrated in extensive simulations based on the New England IEEE-39 bus system \cite{hiskens13}. In this way, the suitability of the technique for PMU measurement with Gaussian and non-Gaussian noise is shown. The applicability to PMU measurements with non-Gaussian noise is especially relevant, as PMU measurement noise has been found to follow a non-Gaussian distribution \cite{Wang18} and the inability to deal with these kind of disturbances is a common shortcoming in KF-based DSE approaches.
\endenu

The remainder of the paper is structured as follows. The mathematical model of the multi-machine power system is introduced in Section~\ref{sec1}. An algebraic observer for the load angle and the quadrature-axis internal voltage is derived in Section~\ref{sec2}. In Section~\ref{sec3}, a DREM-based I\&I  adaptive observer for the relative shaft speed is presented. The proposed method is validated in Section~\ref{sec4} using simulation results.  Final remarks and a brief outlook on future work are given in Section~\ref{sec5}.

\section{Mathematical Model of a Decentralized Multi-Machine Power System}
\label{sec1}
The multi-machine power system is comprised of $N>1$ synchronous generators, each described by the the well-known third-order flux-decay model \cite[Eq. (7.176-7.178)]{SAUetal}, see also \cite[Eq. (11.108)]{MACetal}, \cite[Eq. (3.3), (3.5), (3.15)]{Van_cutsem}. The employed model is in accordance with the Model 1.0 specified in Table 1 on page 17 of the IEEE Guide for Synchronous Generator Modeling Practices and Parameter Verification with Applications in Power System Stability Analyses, Std 1110-2019 \cite{IEEE2020}. 
Thus, the equations for the $i$-th machine read 
\begin{subequations}
	\lab{model}
	\begali{
		\dot x_{1,i}&=\omega_i - \omega_{t,i} = x_{2,i}- \omega_{t,i} + \omega_s ,\\
		\dot x_{2,i}&=\frac{\omega_s}{2H_i}(T_{m,i}-T_{e,i}-D_ix_{2,i}),\\
		\dot x_{3,i} &= \frac{1}{T_{d0,i}'}(-x_{3,i}-(x_{d,i}-x_{d,i}')I_{td,i}+E_{f,i}),
	}
\end{subequations}
where the {\em unknown} state is defined as
\begalis{
	x_i&:=\begmat{ x_{1,i} & x_{2,i} & x_{3,i} \end{bmatrix}^\top=\begin{bmatrix} \delta_i-\theta_{t,i} & \omega_i-\omega_s & E_{q,i}' }^\top,
}
with $\delta_i$ the rotor angle, $\omega_i$ the shaft speed, $\omega_{t,i}$ the terminal voltage speed, $\omega_s$ the nominal synchronous speed, $E_{q,i}'$ the quadrature-axis internal voltage, $T_{e,i}$ the electrical air-gap torque, $E_{f,i}$ the field voltage  and $V_{t,i}$ the terminal voltage magnitude. Furthermore, the constants, which are all assumed {\em unknown}, are the damping factor $D_i$, the inertia constant $H_i$, the mechanical power $T_{m,i}$, the direct-axis transient open-circuit time constant $T_{d0,i}'$, the direct-axis reactance $x_{d,i}$ and the direct-axis transient reactance $x_{d,i}'$. 
Moreover, $\theta_{t,i}$ is the terminal voltage phase angle relative to the global $DQ$-coordinate system rotating at nominal synchronous speed $\omega_s$ and defined as $$\theta_{t,i}=  \theta_{0,i}+\int^t_0  (\omega_{t,i}- \omega_s) d\tau ,$$
where $\theta_{0,i}$ is the initial condition of $\theta_{t,i}$. Hence, the load angle $x_1$ is constant in a synchronized state.

\begrem
The rotor angle $\delta_i$ is the angular position of the rotor with respect to the local synchronously rotating reference frame, while the quadrature-axis internal voltage leads the terminal voltage by the load angle $x_{1,i} = \delta_i-\theta_{t,i}$ \cite{kundur94}. To clarify the difference between the introduced angles and angular speeds, the local $dq$-coordinate system of the $i$-th machine is shown in Figure \ref{fig:angles} in relation to the global $DQ$-coordinate system. 
\endrem
\begin{figure}
	\centering
	\includegraphics[width=0.5\linewidth]{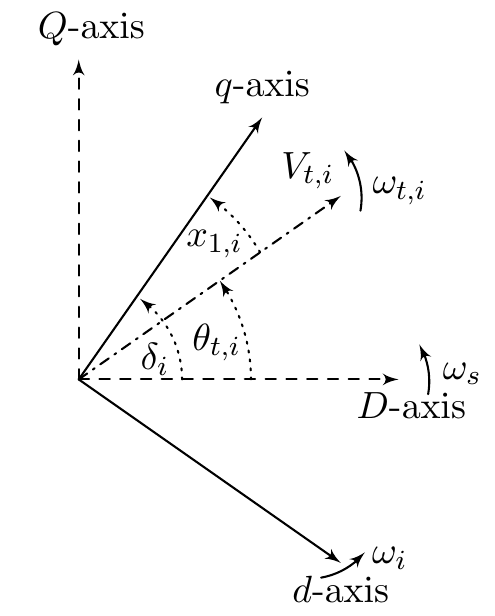}
	\caption{Local $dq$-coordinate system of the $i$-th machine in relation to the global $DQ$-coordinate system.
	}
	\label{fig:angles}
\end{figure}

For the subsequent derivations, we assume the stator resistance to be zero \cite{GHAKAM,GHAKAM2} and make the following assumption on the direct-axis transient reactance $x_{d,i}'$ and the quadrature-axis reactance $x_{q,i}.$
\begin{assumption}
	$x_{d,i}'=x_{q,i}.$ 
	\label{assX}
\end{assumption}
With Assumption~\ref{assX}, the stator equations \cite[Eq. (7.1780-7.181)]{SAUetal} for the $i$-th machine, are given by
\begin{equation}
	\lab{I}
	I_{tq,i}=Y_iV_{t,i}\sin(x_{1,i}),\quad 
	I_{td,i}=Y_i(x_{3,i}-V_{t,i}\cos(x_{1,i})),
\end{equation}
where we have introduced the transient admittance  magnitude $$Y_i =\frac{1}{x_{d,i}'}.$$
With zero stator resistance, the electrical air-gap torque can be approximated by the terminal electrical power (see also \cite{GHAKAM,GHAKAM2})
$$
T_{e,i}\approx P_{t,i}=E_{q,i}'I_{tq,i}=Y_ix_{3,i}V_{t,i}\sin(x_{1,i}),
$$
where we have used \eqref{I} to obtain the last equality. 

Hence, by defining the {\em unknown} constants
\begin{equation}
	\begin{split}
		a_{1,i}&=\frac{\omega_s D_i}{2H_i}, \ \ a_{2,i}= \frac{\omega_s}{2H_i}, \\ a_{3,i} &=\frac{x_{d,i}}{x_{d,i}'T_{d0,i}'}, \ \ a_{4,i}=\frac{x_{d,i}-x_{d,i}'}{x_{d,i}'T_{d0,i}'},
	\end{split}
\end{equation}	
the model \eqref{model} for the $i$-th machine can be compactly written as 
\begin{subequations}
	\lab{x}
	\begali{
		\lab{x1}
		\dot x_{1,i}&=x_{2,i}- \omega_{t,i} + \omega_s,\\
		\lab{x2}	
		\dot x_{2,i}&=-a_{1,i} x_{2,i} + a_{2,i}(T_{m,i} - Y_i V_{t,i} x_{3,i} \sin(x_{1,i})),\\
		\label{x3}
		\dot x_{3,i} &= -a_{3,i} x_{3,i}+ a_{4,i} V_{t,i} \cos(x_{1,i})+ \frac{E_{f,i}}{T_{d0,i}'}.
	}
\end{subequations}
The {\em measurements}, which are provided via a PMU at the generator terminal, are defined as 
\begin{subequations}
	\lab{y}
	\begali{
		y_{1,i}&=V_{t,i}, \label{y1}\\
		y_{2,i}&=P_{t,i}=Y_i V_{t,i} x_{3,i} \sin(x_{1,i}), \label{y2}\\
		y_{3,i}&= Q_{t,i}=\Im\{(V_{td,i}+jV_{tq,i})(I_{d,i}-jI_{q,i})\}\notag \\
		&=Y_i\left( V_{t,i}x_{3,i}\cos(x_{1,i}) -  V_{t,i}^2\right),\label{y3}  \\
		y_{4,i}^2 &= I_{t,i}^2 =I_{td,i}^2+I_{tq,i}^2\notag \\ 
		&=Y_i^2(x^2_{3,i}+V_{t,i}^2-2V_{t,i}x_{3,i} \cos(x_{1,i})),  \label{y4}\\
		y_{5,i} &= f_{t,i} = \frac{\omega_{t,i}}{2 \pi}, \label{y5}
	}
\end{subequations}
where $y_{1,i}>0$ is the terminal voltage, $y_{2,i}$ is the terminal active power, $y_{3,i}$ is the terminal reactive power, $y_{4,i}$ is the terminal current and $y_{5,i}$ the terminal voltage frequency. The imaginary part is denoted by $\Im\{\cdot\}$.

In this note, we first prove that it is possible to algebraically reconstruct the states $x_{1,i}$ and $x_{3,i}$ for the $i$-th machine in a decentralized setting from the terminal measurements 
$$
\bm{y_i}=\begin{bmatrix} y_{1,i} & y_{2,i} & y_{3,i} & y_{4,i} & y_{5,i} \end{bmatrix}^\top.
$$
Then, an adaptive observer for the relative shaft speed $x_{2,i}$ is designed.

\section{An Algebraic Observer for $x_{1,i}$ and $x_{3,i}$}
\label{sec2}
%
As shown in the proposition below, some straightforward algebraic operations on the measured signals $\bm{y_i}$ in \eqref{y} allows us to explicitly compute the unmeasurable states $x_{1,i}$ and $x_{3,i}$, requiring only the knowledge of the transient admittance  magnitude $Y_i$. By construction, this algebraic observer does not require initial conditions, since no differential equation is solved. 

\begin{proposition}\em
	\lab{pro1}
	The states $x_{1,i}$ and $x_{3,i}$ can be determined uniquely from the PMU measurements \eqref{y} via
	\begin{subequations}
		\lab{obs}
		\begin{align}
			\lab{obsx3}
			x_{3,i} &= \sqrt{\frac{y_{4,i}^2 + 2Y_iy_{3,i}}{Y_i^2} + y_{1,i}^2}, \\
			x_{1,i} &= \arcsin\left(\frac{y_{2,i}}{Y_iy_{1,i}x_{3,i}}\right).
			\lab{obsx1}
		\end{align}
	\end{subequations}

\end{proposition}

\begin{proof} 
	From \eqref{y3} and \eqref{y4} it follows that
	$$
	y_{4,i}^2+2Y_iy_{3,i} = Y_i^2 (x^2_{3,i}-y_{1,i}^2).
	$$
	Thus, $x_{3,i}$ can be calculated as
	$$
	x_{3,i} = \sqrt{\frac{y_{4,i}^2 + 2Y_iy_{3,i}}{Y_i^2} + y_{1,i}^2}.
	$$
	With $x_{3,i}$ known, $x_{1,i}$ is obtained from \eqref{y2} as
	$$
	x_{1,i} = \arcsin\left(\frac{y_{2,i}}{Y_iy_{1,i}x_{3,i}}\right).
	$$	
\end{proof}

\begrem
\lab{rem2}
Proposition \ref{pro1} clearly reveals that we can obviate the use of dynamic state observers for the reconstruction of the generators load angle and the quadrature-axis internal voltage from the classical single axis flux-decay model---see Section~\ref{sec5} for a discussion on the potential extension to the more detailed two-axis model \cite{SAUetal}. We underscore the fact that, besides the PMU measurements, the only additional prior knowledge required in Proposition \ref{pro1} is the transient admittance  magnitude $Y_i$. The DSE problem in power systems is now widely recognized to be of major importance to enhance its awareness and security, see \cite{ZHAetal}, hence the interest of this result.  
\endrem

{
	\begrem
	\lab{rem2b}
	Alternatively to the proposed computation of $x_{3,i}$ from \eqref{y3} and \eqref{y4}, a reconstruction from \eqref{y2} and \eqref{y4} is also feasible. By rearranging and squaring of \eqref{y4} we get
	$$(Y_i^2x_{3,i}^2+Y_i^2V_{t,i}^2-y_{4,i}^2)^2 = (Y_i^2V_{t,i}x_{3,i}\cos(x_{1,i}))^2.$$
	Adding the left- and right-hand side to 
	$$(2Y_iy_{2,i})^2 = (2Y_i^2V_{t,i}x_{3,i}\sin(x_{1,i}))^2,$$
	which follows from \eqref{y2}, results after some algebraic manipulations in
	\begalis{
		0 =& (\overbrace{Y_i^2x_{3,i}^2}^{:=a})^2-2\overbrace{Y_i^2x_{3,i}^2}^{=a}(\overbrace{y_{4,i}^2+Y_i^2V_{t,i}^2}^{:=b}) \\&+ \underbrace{4Y_i^2y_{2,i}^2+(Y_i^2V_{t,i}^2-y_{4,i}^4)^2}_{:=c},} 
	yielding an equation of the form 
	$$a^2-2ab+c=0.$$
	Thus, the positive root $x_{3,i}>0$ can be calculated uniquely from
	$$x_{3,i} =\sqrt{\frac{y_{4,i}^2+Y_i^2V_{t,i}^2+2\sqrt{Y_i^2(y_{4,i}^2V_{t,i}^2-y_{2,i}^2)}}{Y_i^2}} ,$$
	since $a>0$ and $c>0$ by definition. 
	\endrem
}

\begrem
\lab{rem3}
In spite of its obvious simplicity the result of Proposition \ref{pro1} has not---to the best of our knowledge---been reported in the literature. A related reference is \cite{UECWED}, where it is shown that the single axis flux-decay model is a differentially flat system whose flat outputs are the network-frame currents. This property is used in  \cite{UECWED} for (open-loop) trajectory planning and feedback linearization of the system. 
\endrem

\section{A DREM-Based I\&I  Adaptive Observer for $x_{2,i}$}
\label{sec3}
%
In this section, we combine the well-established I\&I technique \cite{ASTetal} for observer design with the recently introduced DREM parameter estimator \cite{ARAetal} to design---using $x_{1,i}$ obtained via \eqref{obs}---an adaptive observer for the rotor angular speed $x_{2,i}$ of the $i$-th machine. For the sake of clarity we divide the presentation of the result in three parts: First, an I\&I observer assuming the mechanical parameters $a_{1,i}$, $a_{2,i}$ and $T_{m,i}$ are known. Second, we design a DREM estimator for these parameters. Third, with an {\em ad-hoc} application of the certainty equivalent principle, we propose the final DREM-based I\&I  adaptive observer, replacing the true parameters by their on-line estimates.
\subsection{An I\&I Observer with Known $a_{1,i}$, $a_{2,i}$ and $T_{m,i}$}
\label{subsec31}
%
\begin{lemma}\em
	\lab{lem1}
	Consider the mechanical subsystem dynamics given in \eqref{x1} and \eqref{x2} with $x_{1,i}$ obtained via \eqref{obs}. Define the I\&I observer for the $i$-th machine 
	\begalis{
		\dot{x}^I_{2,i} &= - ({a_{1,i}}+k) ({x}^I_{2,i} + k x_{1,i})+ k(\omega_{t,i}-\omega_s)\\ & \  \ \ + a_{2,i}(T_{m,i} - y_{2,i}), \\
		\hat{x}_{2,i} &= {x}^I_{2,i} + k x_{1,i},
	}   
	with $k>0$ a tuning parameter.  Then, 
	$$
	\tilde {x}_{2,i}(t)=e^{-(a_{1,i}+k)t}\tilde x_{2,i}(0),\;\forall t\geq 0,
	$$
	where $\tilde x_{2,i}:=\hat x_{2,i}-x_{2,i}$ is the state observation error.
\end{lemma}

\begin{proof}
	Following the I\&I observer design technique \cite[Chapter 5]{ASTetal}, we propose to generate the estimate of $x_{2,i}$ as the sum of a proportional and an integral component, with the former being a function of the measurable signals, in this case of $x_{1,i}$. That is,
	$$
	\hat x_{2,i}=x_{2,i}^P(x_{1,i})+x_{2,i}^I,
	$$
	with $x_{2,i}^P(x_{1,i})$ a function to be defined. Computing the time derivative of the observation error $\tilde x_{2,i}$ we get 
	\begalis{
		\dot{\tilde x}_{2,i}&=\dot x_{2,i}^P(x_{1,i})+\dot x_{2,i}^I-\dot x_{2,i}\\
		&={dx_{2,i}^P(x_{1,i})\over dx_{1,i}}(x_{2,i}-\omega_{t,i} +\omega_s) +\dot x_{2,i}^I-\dot x_{2,i}\\
		&={dx_{2,i}^P(x_{1,i})\over dx_{1,i}}(\hat x_{2,i}-\tilde x_{2,i} -\omega_{t,i} +\omega_s)+\dot x_{2,i}^I-\dot x_{2,i}\\ 
		&=k(\hat x_{2,i}-\tilde x_{2,i}-\omega_{t,i} +\omega_s)+\dot x_{2,i}^I+a_{1,i}x_{2,i} \\& \ \  \ - a_{2,i}(T_{m,i} -y_{2,i})\\
		&=- ({a_{1,i}}+k) \tilde {x}_{2,i},
	}  
	where we have selected $x_{2,i}^P(x_{1,i})=kx_{1,i}$ in the fourth identity and used the definition of $\dot{x}^I_{2,i}$ to obtain the last one. Solving the last differential equation completes the proof.
\end{proof}

\begrem
\lab{rem4}
Notice that the design of the I\&I observer does not require the assumption that $T_{m,i}$ is constant, and it may be a time-varying {\em measurable} signal.
\endrem

\subsection{A Parameter Estimator for $a_{1,i}$, $a_{2,i}$ and $T_{m,i}$}
\label{subsec32}
%
The lemma below proposes a DREM-based estimator for the parameters $a_{1,i}$, $a_{2,i}$ and $T_{m,i}$ of the $i$-th machine that ensures convergence under some suitable excitation conditions. 

\begin{lemma}\em
	\lab{lem2}
	Consider the mechanical subsystem dynamics given in \eqref{x1} and \eqref{x2}  with $x_{1,i}$ obtained via \eqref{obs}. Define the vector of unknown parameters 
	\begequ
	\lab{the}
	\bm{\theta_i}:=\col(a_{1,i},a_{2,i},a_{2,i}T_{m,i}),
	\endequ 
	the signals     
	\begali{
		\nonumber
		z_i & := \frac{\lambda^2 p^2}{(p+\lambda)^2}[x_{1,i}] + \frac{\lambda^2 p}{(p+\lambda)^2}[\omega_{t,i}]\\
		\lab{zpsi}
		\bm{\psi_i}&:= \begmat{-\frac{\lambda^2 p}{(p+\lambda)^2}[x_{1,i}] - \frac{\lambda^2 }{(p+\lambda)^2}[\omega_{t,i}-\omega_s]\\ -\frac{\lambda^2}{(p+\lambda)^2}[y_{2,i}]\\ \frac{\lambda^2}{(p+\lambda)^2}[U_{-1}(t)]},
	}   
	with $p:={d \over dt}$, $\lambda>0$ a tuning parameter and $U_{-1}(t)$ a step signal, the vector $\bm{Z_i} \in \rea^3$ and the matrix $\bm{\Psi_i} \in \rea^{3 \times 3}$
	\begali{ \label{hlre}
		\nonumber
		\bm{Z_i} &:=\bm{\mathcal{H}}[z_i]\\
		\bm{\Psi_i} &:= \bm{\mathcal{H}}[\bm{\psi_i}^\top],
	}
	with $\bm{\mathcal{H}}$ a {\em linear, single-input 3-output, bounded-input bounded-output (BIBO)-stable operator} and the signals
	\begali{
		\nonumber
		\bm{\mathcal{Z}_i}&:=\adj\{\bm{\Psi_i}\} \bm{Z_i}\\
		\lab{calzdel}
		\Delta_i & :=\det\{\bm{\Psi_i}\},
	}
	where $\adj\{\cdot\}$ is the adjunct (also called ``adjugate") matrix and $\det\{\cdot\}$ is the determinant.
	
	The scalar parameter estimators 
	\begequ
	\lab{dothatthe}
	\dot {\hat \theta}_i^j=-\gamma_i^j \Delta_i(\Delta_i  \hat \theta_i^j -\mathcal{Z}_i^j),\;j=1,2,3,
	\endequ
	with $\gamma_i^j>0$ adaptation gains, ensures the parameter estimation error $\tilde \theta_i^j:=\hat \theta_i^j -\theta_i^j$ verifies 
	$$
	\liminf \tilde \theta_i^j(t)=0,\;j=1,2,3,
	$$
	provided $\Delta_i \notin \mathcal{L}_2$, that is,
	\begequ
	\lab{exccon}
	\liminf \int_0^t \Delta_i^2(s)ds=\infty.
	\endequ
\end{lemma}

\begin{proof}
	From \eqref{x1} and \eqref{x2} we get
	$$
	\ddot x_{1,i} +\dot{\omega}_{t,i}=-a_{1,i} (\dot x_{1,i} +\omega_{t,i}-\omega_s) + a_{2,i}(T_{m,i} - y_{2,i}).
	$$
	Applying the filter $\frac{\lambda^2}{(p+\lambda)^2}$ to the equation above we get 
	\begin{equation}
		\begin{split}
			\lab{lre0}
			&\frac{\lambda^2 p^2}{(p+\lambda)^2}[x_{1,i}] + \frac{\lambda^2 p}{(p+\lambda)^2}[\omega_{t,i}]=-a_{1,i}\frac{\lambda^2 p}{(p+\lambda)^2}[x_{1,i}] \\ & \ \ \ \ - a_{1,i}\frac{\lambda^2 }{(p+\lambda)^2}[\omega_{t,i}-\omega_s] - a_{2,i} \frac{\lambda^2}{(p+\lambda)^2}[y_{2,i}]\\ & \ \ \ \ +a_{2,i} T_{m,i}\frac{\lambda^2}{(p+\lambda)^2}[U_{-1}(t)].
		\end{split}
	\end{equation}   
	Using \eqref{the} and \eqref{zpsi} we can write \eqref{lre0} as a linear regression equation
	\begequ
	\lab{lre}
	z_i =  \bm{\psi_i}^\top  \bm{\theta_i}.
	\endequ
	Following the DREM procedure \cite{ARAetal,ORTetal} we carry out the next operations utilizing \eqref{lre} 
	\begalis{
		\bm{\mathcal{H}}[z_i] &=  \bm{\mathcal{H}}[\bm{\psi_i}^\top  \bm{\theta_i}]
		\qquad \qquad \quad\;(\Leftarrow\;\bm{\mathcal{H}}[ \cdot])\\
		\bm{Z_i} &= \bm{\Psi_i} \bm{\theta_i} \qquad \qquad \;\qquad(\Leftrightarrow\;\eqref{hlre})\\
		\adj\{\bm{\Psi_i}\} \bm{Z_i} &= \adj\{\bm{\Psi_i}\} \bm{\Psi_i \theta_i} \qquad \quad\; (\Leftarrow\;\adj\{\bm{\Psi_i}\} \times)\\
		\mathcal{Z}_i^j &= \Delta_i \theta_i^j,\;j=1,2 ,3\quad \quad\;(\Leftrightarrow\;\eqref{calzdel}).
	}
	where, to obtain the second identity, we have used the linearity of the operator $\bm{\mathcal{H}}$ and for the last identity the fact that for any (possibly singular) $q \times q$ matrix $\bm{M}$ we have $\adj\{\bm{M}\} \bm{M}=\det\{\bm{M}\}\bm{I_q}$. Replacing the latter equation in \eqref{dothatthe} yields the error dynamics
	$$
	\dot {\tilde \theta}_i^j=-\gamma_i^j \Delta_i^2  \tilde \theta_i^j,\;j=1,2,3.
	$$
	The proof is completed observing that the solutions of the later equations are given by
	$$
	\tilde \theta_i^j(t)=e^{-\gamma_i^j \int_0^t \Delta_i^2(\tau)d\tau}\tilde \theta_i^j(0),\;j=1,2,3.
	$$
\end{proof}

\begrem
\lab{rem5}
It is clear that it is possible to directly apply a classical gradient descent estimator to the {\em vector} linear regression equation \eqref{lre}, that is
\begequ
\lab{graest}
\bm{\dot {\hat {\theta}}_i}=-\bm{\Gamma_i \psi_i}(\bm{\psi_i}^\top   \bm{\hat \theta_i} -z_i),\;\bm{\Gamma_i}>0,
\endequ
which yields the error equation
\begequ
\lab{errequ}
\bm{\dot {\tilde {\theta}}_i}=-\bm{\Gamma_i \psi_i\psi_i}^\top   \bm{\tilde \theta_i}.
\endequ
Our motivation to use, instead, the more complicated DREM estimator is to relax the excitation assumptions that guarantee its convergence. Indeed, it is well-known  \cite[Theorem 2.5.1]{SASBOD} that a necessary and sufficient conditions for global (exponential) convergence of the error equation \eqref{errequ} is that the regressor $\bm{\psi_i}$ satisfies a stringent {\em persistent excitation} requirement  \cite[Equation 2.5.3]{SASBOD}. Some simulation results have shown that this condition is not satisfied in normal operation of the power system. On the other hand, it has been shown in \cite{ORTetal} that the non-square-integrability condition \eqref{exccon} is {\em strictly weaker} than persistent excitation. 
\endrem

\begrem
\lab{rem6}
We have presented Lemma \ref{lem2}  for the scenario where the mechanical power $T_{m,i}$ is constant but  {\em unknown}. It is clear that it is straightforward to extend it---applying the filter  $\frac{\lambda^2}{(p+\lambda)^2}$ and redefining $z_i$ and the regressor vector $\bm{\psi_i}$---to the case where $T_{m,i}$ is time-varying, but {\em measurable}. 
\endrem

\subsection{Adaptive I\&I Observer}
\lab{subsec33}
%
Combining the known parameter observer of Lemma \ref{lem1} with the parameter estimator of Lemma \ref{lem2} yields the final certainty equivalent adaptive observer
\begin{equation}
	\begin{split} 
		\dot{x}^I_{2,i} &= - ({\hat \theta_{1,i}}+k) ({x}^I_{2,i} + k x_{1,i}) +k(\omega_{t,i}-\omega_s) \\ & \ \ \ -\hat \theta_{2,i} y_{2,i} + \hat \theta_{3,i},\\
		\hat{x}_{2,i} &= {x}^I_{2,i} + k x_{1,i},
	\end{split} 
	\label{eq:obsx2}
\end{equation}
with $\bm{\hat \theta_i}$ defined via \eqref{zpsi}-\eqref{dothatthe}. Under the excitation assumption \eqref{exccon}, convergence of the adaptive observer is established via standard cascaded systems stability analysis, see {\em e.g.}, \cite{VID}. 
%
\section{Simulation Results}
\label{sec4}
In this section, we present simulation results demonstrating the effectiveness of the proposed methods. 
We use the well-known New England IEEE 39 bus system shown in Figure \ref{fig:IEEE39Bus}, with the parameters provided in \cite{hiskens13}. All synchronous generators are represented by the third-order flux-decay model \eqref{eq:full_sg} and are equipped with automatic voltage regulators (AVRs) and power system stabilizers (PSSs) according to \cite{hiskens13}. Thus, each generator is represented by the following 9-dimensional model:

\begin{equation}
	\label{eq:full_sg}
	\begin{split}
		\dot x_{1,i}&=x_{2,i}- \omega_{t,i} + \omega_s,\\
		\dot x_{2,i}&=-a_{1,i} x_{2,i} + a_{2,i}(T_{m,i} - Y_i V_{t,i} x_{3,i} \sin(x_{1,i})),\\
		\dot x_{3,i} &= -a_{3,i} x_{3,i}+ a_{4,i} V_{t,i} \cos(x_{1,i})+ \frac{E_{f,i}}{T_{d0,i}'},\\
		\dot{V}_{f,i} &= \frac{1}{T_{R,i}} (V_{t,i}-V_{f,i}),\\
		\dot{q}_{i} &= \frac{1}{T_{B,i}} \left(\left(1-\frac{T_{C,i}}{T_{B,i}}\right)(V_{\text{ref},i} - V_{f,i}+ V_{\text{pss},i})-q_i\right),\\
		\dot{E}_{f,i} &= \frac{1}{T_{A,i}} \left(K_{A,i}\left(q+\frac{T_{C,i}}{T_{B,i}}\right)(V_{\text{ref},i} - V_{f,i}+ V_{\text{pss},i})-E_{f,i}\right),\\
		\dot{p}_{1,i} &= -c_{1,i}p_{1,i} + p_{2,i} + (c_{4,i} -c_{1,i}c_{3,i})x_{2,i},\\
		\dot{p}_{2,i} &= -c_{2,i}p_{1,i} + p_{3,i} + (c_{5,i} -c_{2,i}c_{3,i})x_{2,i},\\
		\dot{p}_{3,i} &= -p_{1,i} -c_{1,i}c_{3,i}x_{2,i},\\
		V_{\text{pss},i} &=p_{1,i} + c_{3,1} x_{2,i}, 
	\end{split}
\end{equation}
with 
\begin{equation}
	\begin{split}
		c_{1,i} & = \frac{T_{4,i}T_{w,i}+T_{4,i}T_{2,i} + T_{2,i}T_{w,i} }{T_{w,i} T_{4,i} T_{2,i}} ,\\
		c_{2,i} & = \frac{T_{w,i} +T_{4,i} +T_{2,i}}{T_{w,i} T_{4,i} T_{2,i}},\\
		c_{3,i} & = \frac{K_{p,i} T_{1,i} T_{3,i}} {T_{2,i} T_{4,i}},\\
		c_{4,i} & = \frac{K_{p,i} (T_{1,i}+ T_{3,i})} {T_{2,i} T_{4,i}},\\
		c_{5,i} & =\frac{K_{p,i}}{T_{2,i} T_{4,i}},
	\end{split}
\end{equation}
where the differential equations for the AVR and PSS are taken from \cite[Figure 2,3]{hiskens13}. The signals $V_{f,i}$, $q_i$ and $p_{j,i}, j = 1,2,3,$ are intermediate variables required for the AVR and PSS respectively. All time constants and gains for the AVR and PSS are defined in \cite{hiskens13}. The employed parameters for the DREM-based I\&I adaptive observer for $x_{2,i}$ are given in Table \ref{tab:params}, where we defined the operator $\bm{\mathcal{H}}$ as a stable transfer matrix
\begin{equation}
	\bm{\mathcal{H}}(s) := \begin{bmatrix} 1 &e^{-sd_1} & \frac{s+k_1}{s+k_2}e^{-sd_2}\end{bmatrix}^\top.
\end{equation}
where $s$ is the Laplace variable. All simulations are performed using MATLAB.

To monitor the system, we assume a PMU installed at the terminal bus of generator 5. Albeit, monitoring any or all other generators in the system is feasible employing the proposed method and assuming additional PMUs installed at the terminal buses of the generators to be monitored. To validate the performance under realistic operation conditions three cases are considered: 
\begin{enumerate}
	\item A nominal, {\em i.e.} noisefree, case.
	\item A case with zero mean Gaussian noise added to the PMU measurements. The signal to noise ratio (SNR) is set to 45 dB, which was identified as a good approximation of noise power by analyzing real PMU data in \cite{Brown16}.
	\item A case with zero mean Laplacian noise added to the PMU measurements. The signal to noise ratio is set to 45 dB. In \cite{Wang18}, Laplacian noise was recommended to simulate realistic PMU measurement errors.
\end{enumerate}
To mimic realistic PMUs, the sampling frequency of the measurements is set to 60 [Hz] in all scenarios, which corresponds to the nominal system frequency and lies within the typical PMU sampling rate of 10 to 120 [Hz] \cite{zhou2015}.

\begin{table}
	\centering
	\begin{tabular}{l|l|l}
		\cmidrule{1-3}
		Symbol & Description& Value  \\  	\cmidrule{1-3}
		$k_1$ & Design parameter & 6\\
		$k_2$ & Design parameter & 4\\
		$d_1$ & Delay constant & 4\\
		$d_2$ & Delay constant & 1\\
		$\gamma_{1,2,3}$ &Adaptation gains (DREM) & $1.5\cdot10^{7}$ \\
		$\lambda$ & Filter parameter& 0.5\\
		$k$ & Observer gain& 1\\
		\cmidrule{1-3}
	\end{tabular}
	\caption{Parameters for the DREM-based I\&I  adaptive observer.}
	\label{tab:params}
\end{table} 

\begin{figure}
	\centering
	\includegraphics[width=1\linewidth]{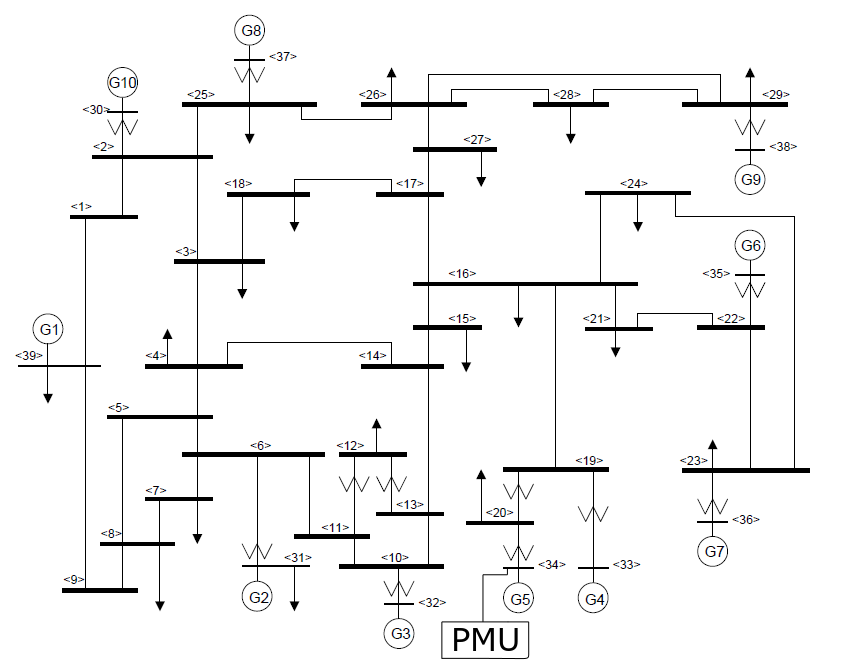}
	\caption{New England IEEE 39 bus system (figure taken from \cite{hiskens13}).}
	\label{fig:IEEE39Bus}
\end{figure}

The performance evaluation is undertaken as follows. For all three considered cases, we demonstrate that $x_{1,i}$ and $x_{3,i}$ can be reconstructed instantaneously using \eqref{obs} in Proposition \ref{pro1}. Due to the algebraic structure initial conditions are not required. Then, we show that already minor load variations, as continuously occurring during regular operation of the power system, provide sufficient excitation to estimate the unknown parameters following Lemma~\ref{lem2}. Thus, in combination with Lemma~\ref{lem1}, the second state $x_{2,i}$ can be reconstructed via the certainty equivalence adaptive observer \eqref{eq:obsx2}. In addition, we show that the proposed observers \eqref{obs} and \eqref{eq:obsx2} are capable of capturing the fast transient behavior of the system after a three-phase short circuit.
To quantify the performance of the state estimation, we compute the symmetric mean absolute percentage error (sMAPE) for each case by comparing the estimated and real values of the states $x_{1,5}$, $x_{2,5}$ and $x_{3,5}$. The sMAPE is defined as \cite{armstrong78}
\begin{equation}
	\text{sMAPE} = \frac{100 \%}{M} \sum_{k=1}^{M} \frac{|\hat{x}_{j,i}-x_{j,i}|}{(|\hat{x}_{j,i}|+|x_{j,i}|)/2} ,\;j=1,2,3,
\end{equation}
where $M$ is the number of considered data points.

\subsection{Scenario 1: Load Variations}
We simulated minor load variations in the system. The resulting frequency variations are within $60\pm0.05$~[Hz] and hence consistent with those during regular operation of transmission grids \cite{weissbach09}.
The proposed method is tested employing the three different cases concerning the disturbance of the PMU measurements. For computing the sMAPE of $x_{2,5}$ only the time with converged parameters, videlicet after 50~seconds, is considered. The results are shown in Table \ref{tab:wape_load}.
\begin{table}
	\centering
	\begin{tabular}{l|l|l|l}
		\cmidrule{1-4}
		Scenario 1&Case& State &sMAPE [\%] \\ 			\cmidrule{1-4}
		&  & $x_{1,5}$ & 0 \% \\
		&1 & $x_{2,5}$ & 0.03 \% \\	
		& & $x_{3,5}$ & 0 \% \\		\cmidrule{2-4}
		Load&  & $x_{1,5}$ & 0.13 \% \\
		variations	& 2 & $x_{2,5}$ &0.98 \% \\	
		&  & $x_{3,5}$ &  0.06 \% \\	\cmidrule{2-4}
		&  & $x_{1,5}$ & 0.11 \% \\
		& 3 & $x_{2,5}$ & 1.11\ \% \\	
		&  & $x_{3,5}$ & 0.05 \% \\	\cmidrule{1-4}
	\end{tabular}
	\caption{sMAPE of the state observer during load variations.}
	\label{tab:wape_load}
\end{table}
The simulation results of the algebraic observer introduced in Proposition \ref{pro1}, are shown in Figures \ref{fig:x_g5_load_no_noise}, \ref{fig:x_g5_load_gaussian} and \ref{fig:x_g5_load_laplace}. Since the reconstruction of the states $x_{1,i}$ and $x_{3,i}$ does not require the solution of a differential equation, initial conditions are not necessary. Hence, in the noisefree scenario (Case 1) those states can be reconstructed instantaneously with a sMAPE of 0~\%. In the disturbed scenarios (Case 2 and 3) a very small estimation error is observed. This is seen in a sMAPE of below 0.14 \% in the cases with Gaussian as well as Laplacian measurement noise.  

\begin{figure}
	\centering
	\includegraphics[width=1\linewidth]{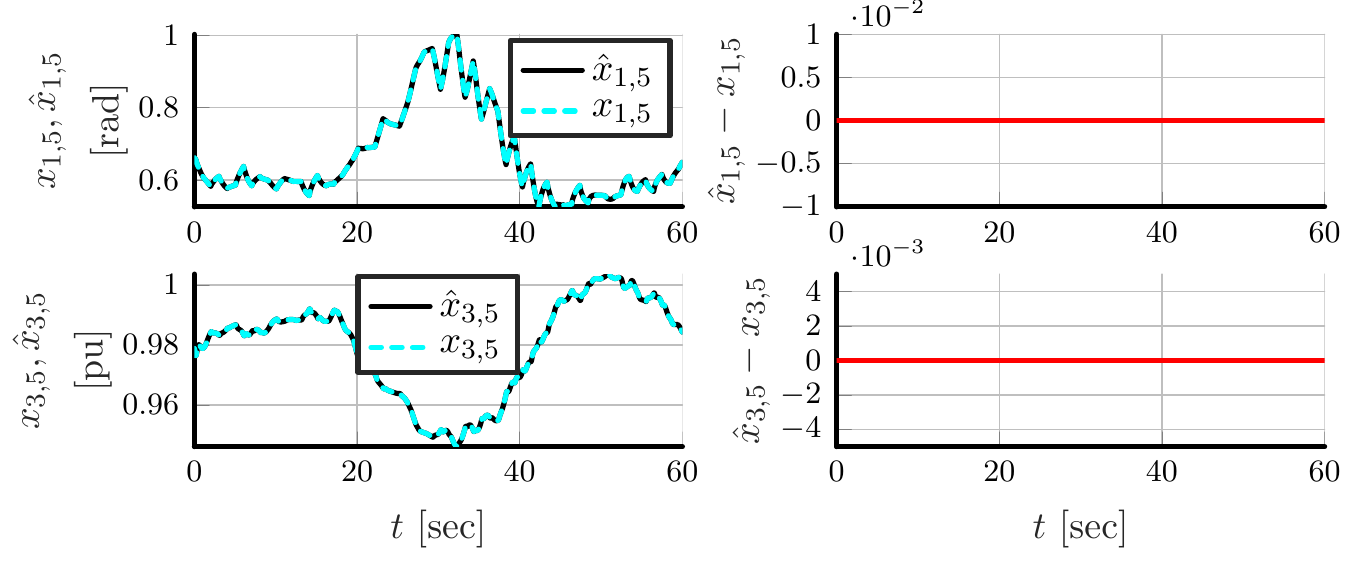}
	\caption{Scenario 1, Case 1: Algebraic state estimation for $x_{1,5}$ and $x_{3,5}$ of generator 5 in the presence of load variations with noisefree PMU measurements. Due to the algebraic structure of this reconstruction no initial conditions are required.}
	\label{fig:x_g5_load_no_noise}
\end{figure}
\begin{figure}
	\centering
	\includegraphics[width=1\linewidth]{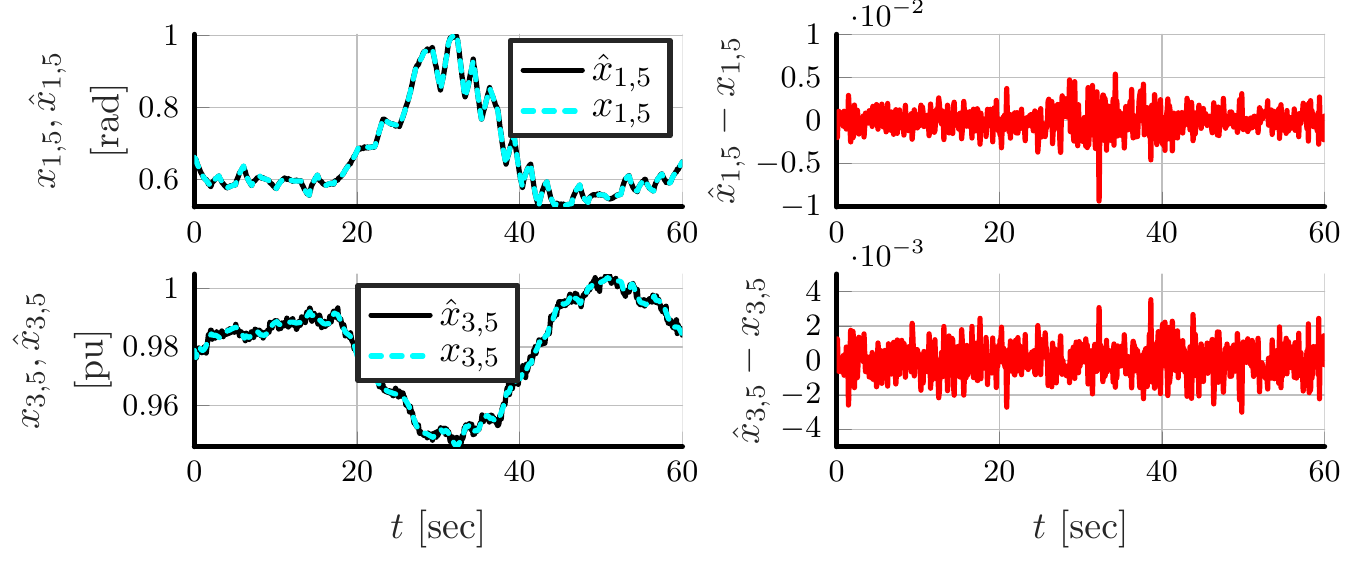}
	\caption{Scenario 1, Case 2: Algebraic state estimation for $x_{1,5}$ and $x_{3,5}$ of generator 5 in the presence of load variations. The PMU measurements are disturbed by zero mean Gaussian noise with a SNR of 45 dB. Due to the algebraic structure of this reconstruction no initial conditions are required.}
	\label{fig:x_g5_load_gaussian}
\end{figure}

\begin{figure}
	\centering
	\includegraphics[width=1\linewidth]{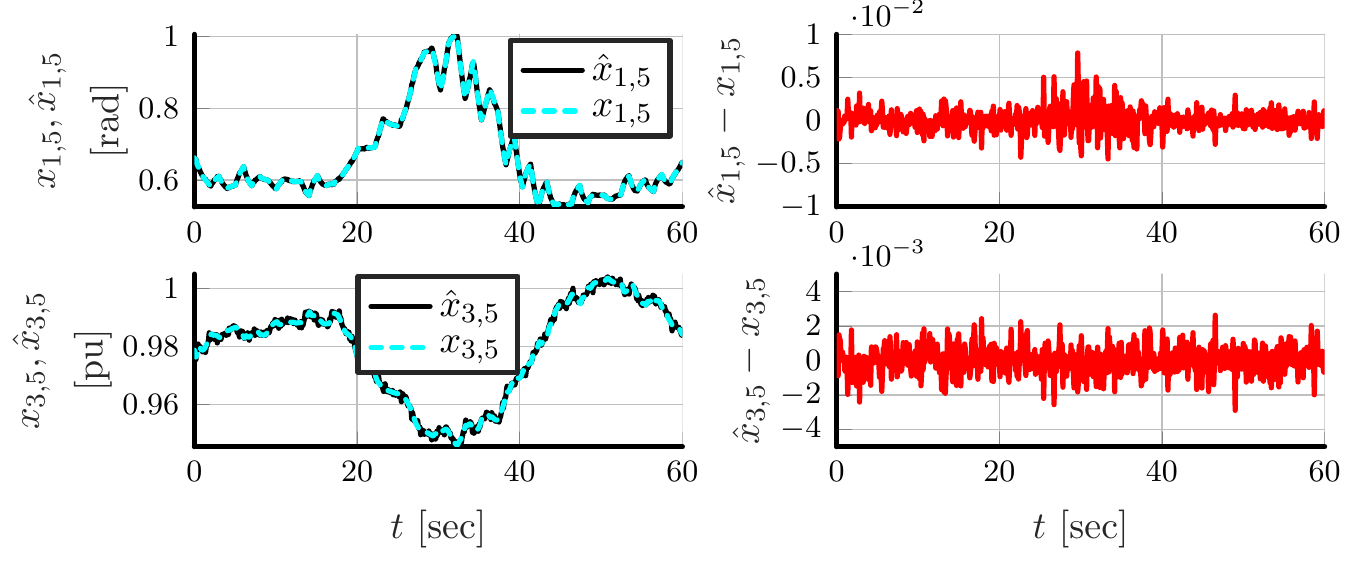}
	\caption{Scenario 1, Case 3: Algebraic state estimation for $x_{1,5}$ and $x_{3,5}$ of generator 5 in the presence of load variations. The PMU measurements are disturbed by zero mean Laplacian noise with a SNR of 45 dB. Due to the algebraic structure of this reconstruction no initial conditions are required.}
	\label{fig:x_g5_load_laplace}
\end{figure}

In Figures \ref{fig:params_g5_load_no_noise}, \ref{fig:params_g5_load_gaussian} and \ref{fig:params_g5_load_laplace}, the results of the DREM-based parameter estimation (Lemma \ref{lem2}) and the I\&I adaptive observer for $x_{2,i}$ (Lemma \ref{lem1}) are shown. The y-axis of these plots is limited to the most relevant range. It can be seen that the simulated load variations provide enough excitation for the DREM-based parameter estimation to converge towards the real values in all three cases. As stated in Lemma~\ref{lem1}, the I\&I adaptive observer for $x_{2,i}$ depends on the estimated parameters. Thus, in the noisefree case the observer error $\hat x_{2,i}-x_{2,i}$ only converges towards zero after the parameter estimates $\bm{\hat\theta_i}$ of $\bm{\theta_i}$ in \eqref{the} have converged at $t=50$~secs. In the cases with noisy PMU measurements, a very small state estimation error is observed even after $t=50$~secs. However, the sMAPE of $x_{2,i}$ is below 1.12~\% for the case with Gaussian as well as Laplacian measurement noise, which are very good results. 

\begin{figure}
	\centering
	\includegraphics[width=1\linewidth]{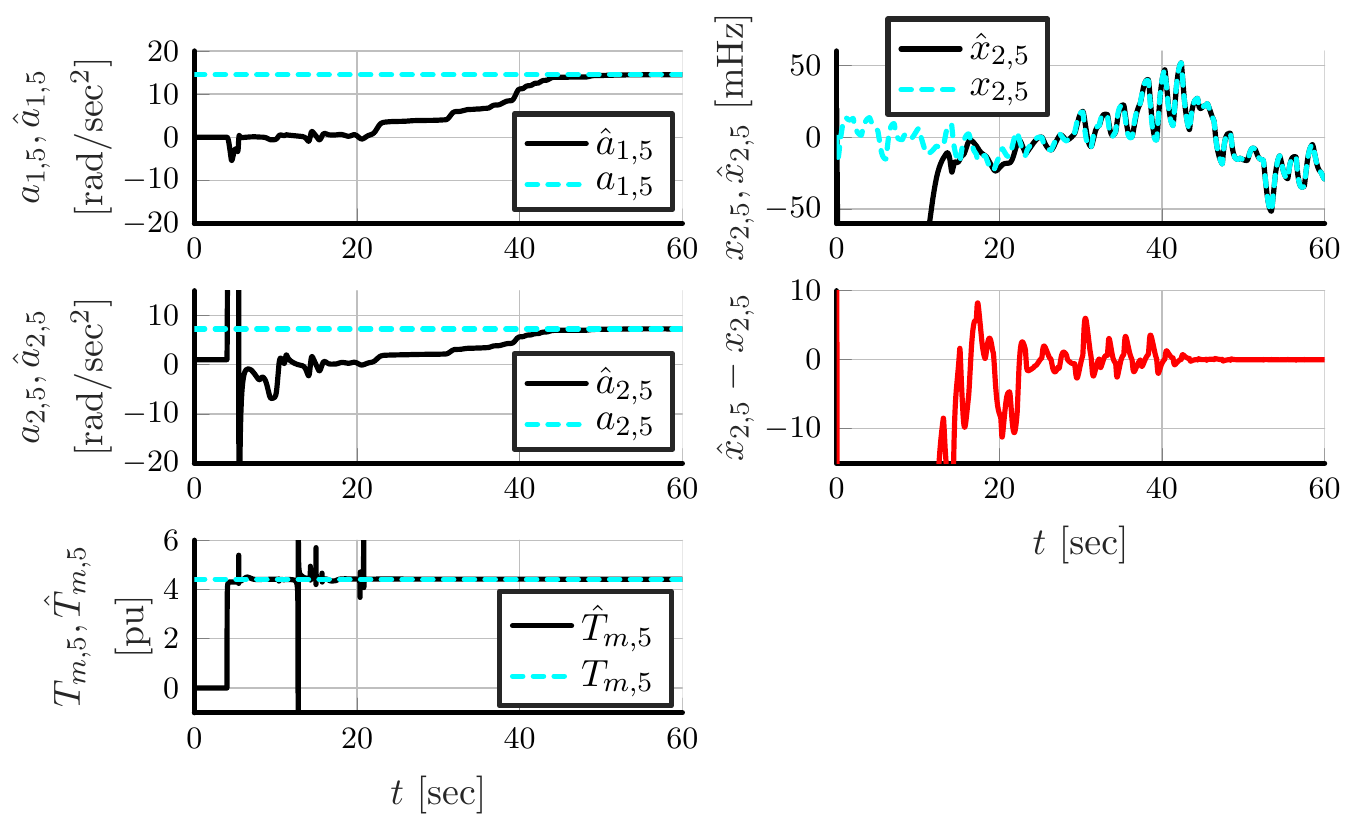}
	\caption{Scenario 1, Case 1: Simulation results of the DREM-based parameter estimation and the I\&I adaptive observer for $x_{2,5}$ of generator 5 in presence of load variations with noisefree PMU measurements.}
	\label{fig:params_g5_load_no_noise}
\end{figure}

\begin{figure}
	\centering
	\includegraphics[width=1\linewidth]{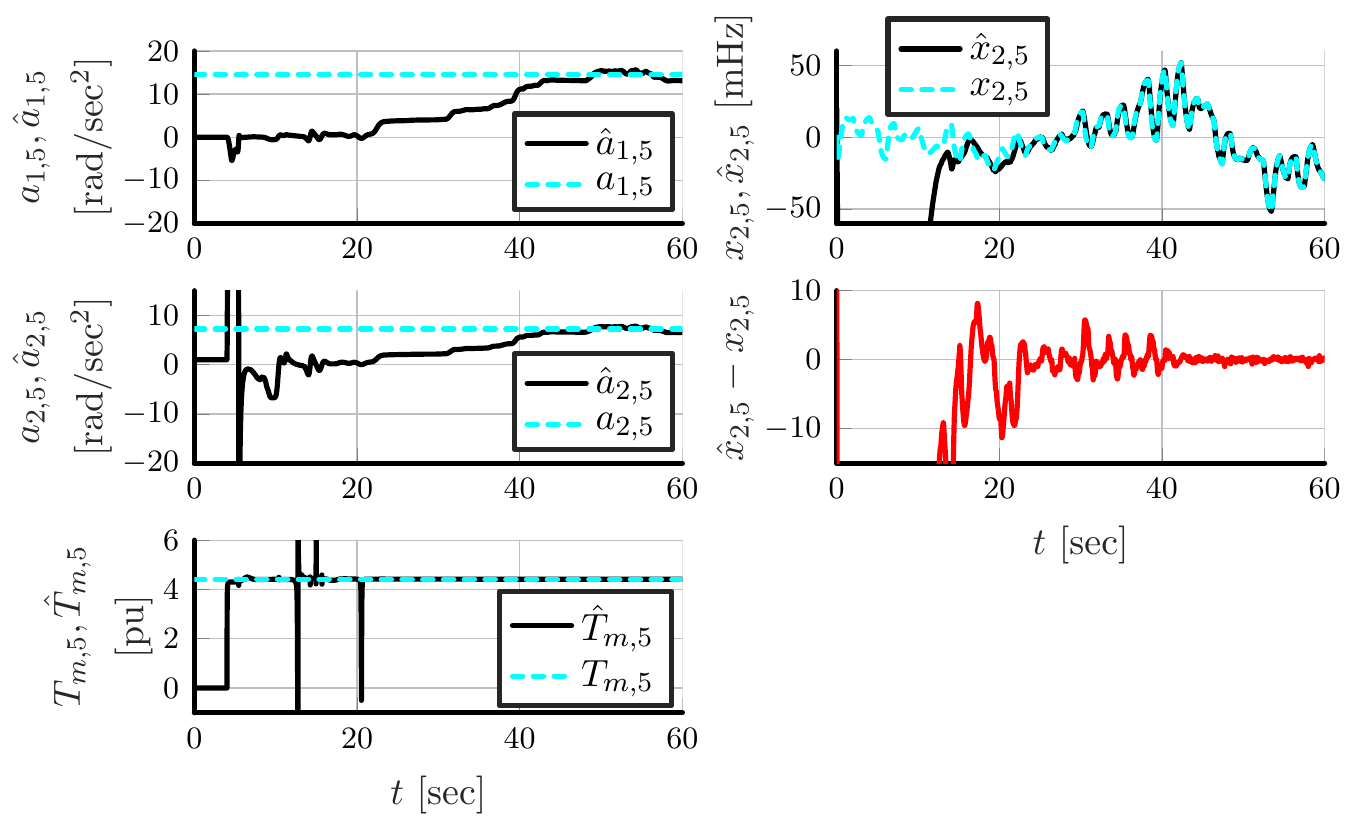}
	\caption{Scenario 1, Case 2: Simulation results of the DREM-based parameter estimation and the I\&I adaptive observer for $x_{2,5}$ of generator 5 in presence of load variations. The PMU measurements are disturbed by zero mean Gaussian noise with a SNR of 45 dB.}
	\label{fig:params_g5_load_gaussian}
\end{figure}

\begin{figure}
	\centering
	\includegraphics[width=1\linewidth]{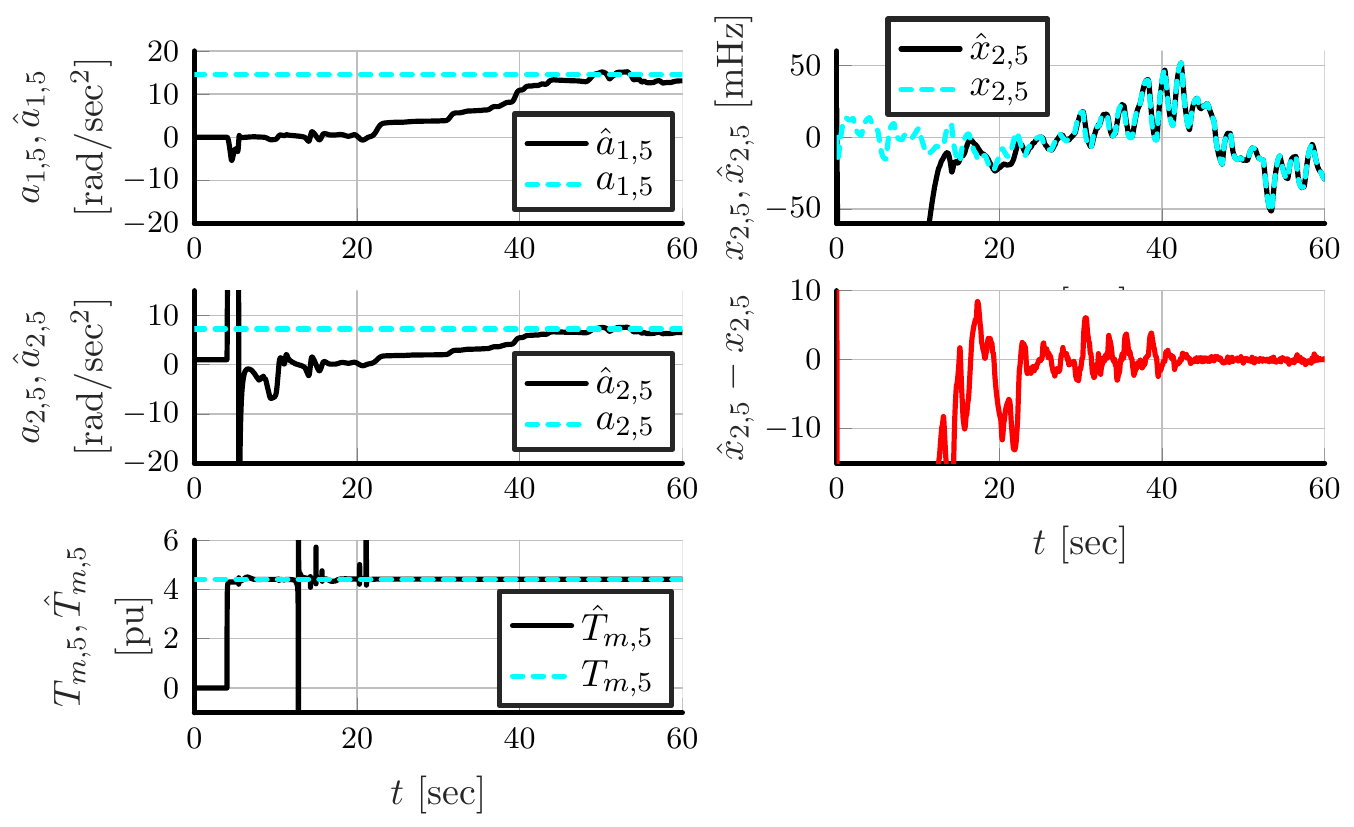}
	\caption{Scenario 1, Case 3: Simulation results of the DREM-based parameter estimation and the I\&I adaptive observer for $x_{2,5}$ of generator 5 in presence of load variations. The PMU measurements are disturbed by zero mean Laplacian noise with a SNR of 45 dB.}
	\label{fig:params_g5_load_laplace}
\end{figure}

\subsection{Scenario 2: Three-Phase Short Circuit}
In this scenario, following \cite{hiskens13} we simulated a three-phase short circuit at Bus 16 occurring at $t = 2$ secs and cleared at $t=2.2$ secs. It is assumed that the DREM-based parameter estimation has already converged during regular operation before the fault. Thus, the parameters are assumed known in this scenario. The performance of the state estimation is quantified by computing the sMAPE. The results are shown in Table \ref{tab:wape_load}. To avoid a division by zero in the calculation of the sMAPE of $x_{2,5}$, only the transient period between $t \approx 2$ secs and $t \approx 3.5$ secs is considered. 

The results of the algebraic and dynamic state estimation are shown in Figures \ref{fig:x_g5_sc_no_noise}, \ref{fig:x_g5_sc_gaussian} and \ref{fig:x_g5_sc_laplace} for the monitored generator 5. As is to be expected, the algebraic observer \eqref{obs} exhibits no estimation error in the noiseless case. In the cases 2 and 3, the states $x_{1,5}$ and $x_{3,5}$ are reconstructed with very high accuracy and a sMAPE of below 0.13~\%, which is in the same range as observed in the simulations of load variations. Furthermore, the observer \eqref{eq:obsx2} for the frequency $x_{2,i}$ also performs satisfactorily, with only a minor estimation error shortly after the fault in all three cases. The computed sMAPE is below 9~\%, which is a good result considering the large and rapid transient deviation after the three-phase short circuit.

\begin{table}
	\centering
	\begin{tabular}{l|l|l|l}
		\cmidrule{1-4}
		Scenario 2	&Case& State &sMAPE [\%] \\ 			\cmidrule{1-4}
		&  & $x_{1,5}$ & 0 \% \\
		&1 & $x_{2,5}$ &  4.92 \% \\	
		& & $x_{3,5}$ & 0 \% \\		\cmidrule{2-4}
		Short&  & $x_{1,5}$ &  0.12 \% \\
		circuit	& 2 & $x_{2,5}$ &  8.58 \% \\	
		&  & $x_{3,5}$ &   0.06 \% \\	\cmidrule{2-4}
		&  & $x_{1,5}$ &  0.1 \% \\
		& 3 & $x_{2,5}$ &  8.94 \% \\	
		&  & $x_{3,5}$ &  0.05 \% \\	\cmidrule{1-4}
	\end{tabular}
	\caption{sMAPE of the state observer during a three-phase short circuit.}
	\label{tab:wape_sc}
\end{table}
\begin{figure}
	\centering
	\includegraphics[width=1\linewidth]{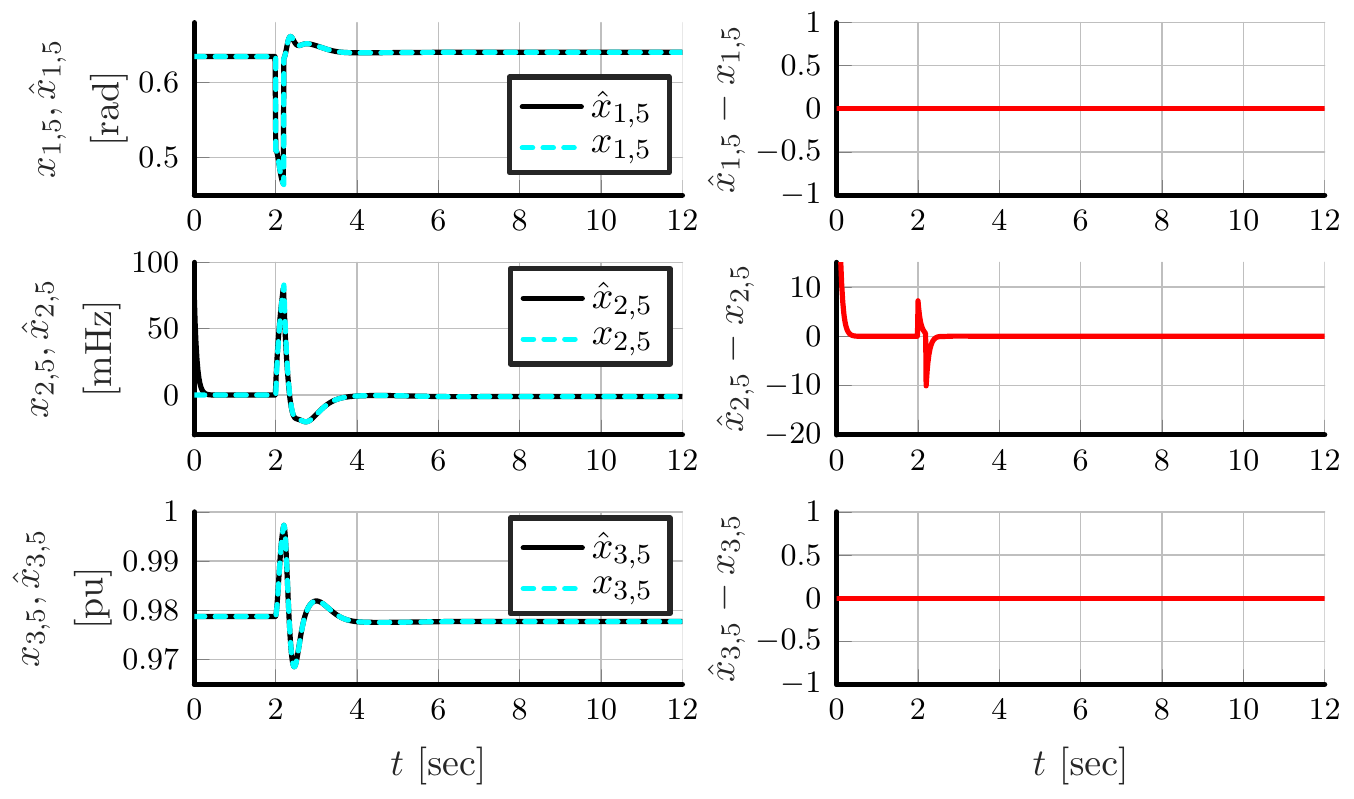}
	\caption{Scenario 2, Case 1: State estimation results for generator 5 during a three-phase short circuit at Bus 16 with noisefree PMU measurements. For the algebraic reconstruction of $x_{1,5}$ and $x_{3,5}$ no initial conditions are required.}
	\label{fig:x_g5_sc_no_noise}
\end{figure}

\begin{figure}
	\centering
	\includegraphics[width=1\linewidth]{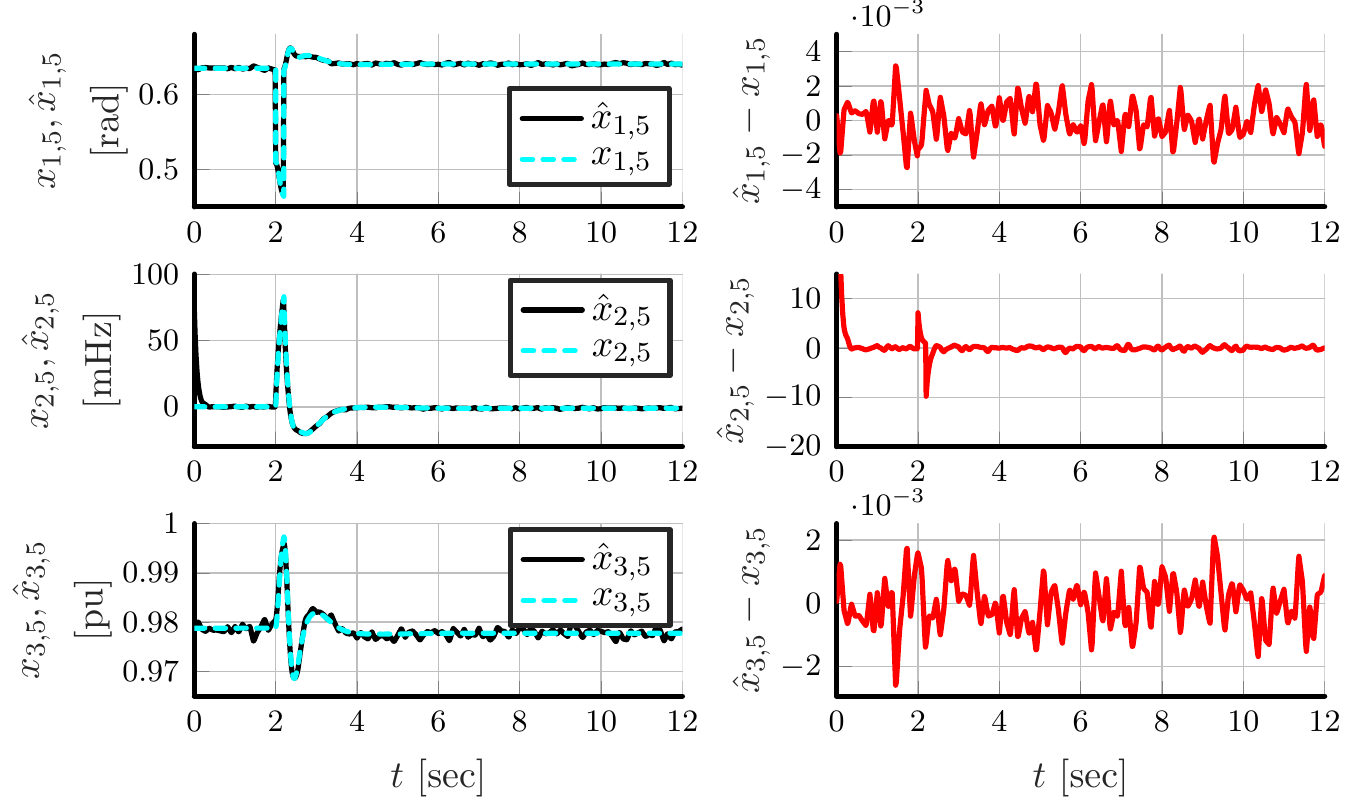}
	\caption{Scenario 2, Case 2: State estimation results for generator 5 during a three-phase short circuit at Bus 16. The PMU measurements are disturbed by zero mean Gaussian noise with a SNR of 45 dB. For the algebraic reconstruction of $x_{1,5}$ and $x_{3,5}$ no initial conditions are required.}
	\label{fig:x_g5_sc_gaussian}
\end{figure}

\begin{figure}
	\centering
	\includegraphics[width=1\linewidth]{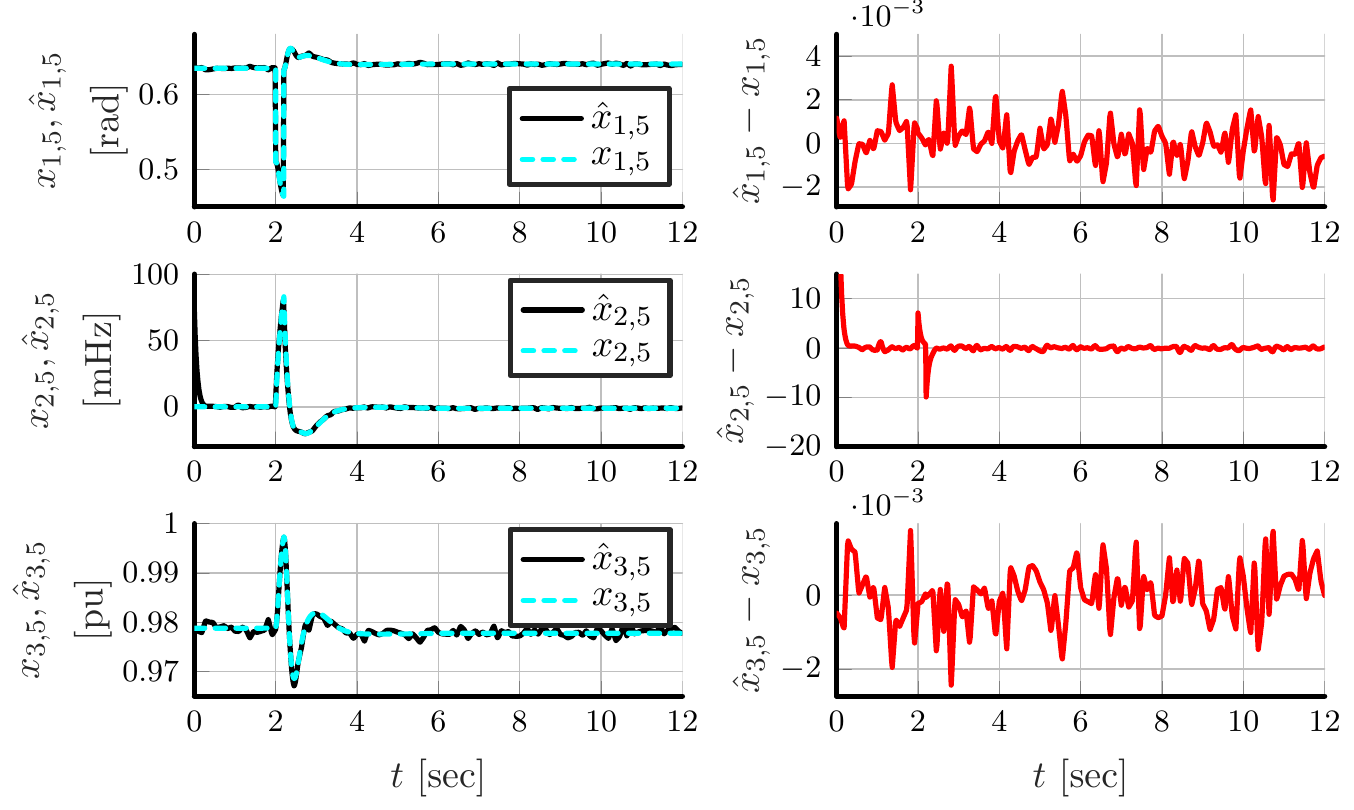}
	\caption{Scenario 2, Case 3: State estimation results for generator 5 during a three-phase short circuit at Bus 16. The PMU measurements are disturbed by zero mean Laplace noise with a SNR of 45 dB. For the algebraic reconstruction of $x_{1,5}$ and $x_{3,5}$ no initial conditions are required.}
	\label{fig:x_g5_sc_laplace}
\end{figure}

\section{Conclusions and Future Research}
\label{sec5}
A decentralized mixed algebraic and dynamic state observer was presented for DSE in multi-machine power systems. It was shown that the load angle and the quadrature-axis internal voltage can be reconstructed algebraically from available PMU measurements at the terminal bus of a synchronous generator. For observing the relative shaft speed a DREM-based I\&I adaptive observer was proposed. 

In simulation studies using the New England IEEE 39 bus system the effectiveness of the proposed observer was demonstrated, using realistic PMU measurements sampled at 60 [Hz] and disturbed by Gaussian and Laplacian noise. In particular, the convergence of the DREM-based parameter estimation was shown under regular operation conditions and load variations. Moreover, by simulating a three-phase short circuit the ability of the method to monitor the state evolution during fast transients was demonstrated.

The classical flux-decay model \eqref{model} provides a fairly accurate description of the behavior of a synchronous machine to assess transient stability in a multi-machine scenario. However, Assumption \ref{assX} can be restrictive in some scenarios and it has recently been argued that its precision can be improved by including additional dynamic effects. For instance, it is argued in \cite[Chapter 11]{MACetal} that including a second differential equation to account for rotor body effects in the $q$-axis significantly improves the accuracy of the model. This leads to a fourth-order model. Our current research is aimed at extending Proposition \ref{pro1}, {\em i.e.}, the reconstruction via algebraic operations of (some of) the system's state variables from the PMU measurements, for the third-order model without relying on Assumption \ref{assX} and for the fourth-order model. Unfortunately, it is possible to show that for the fourth-order model the (relevant part) of the mapping $y \mapsto x$ does not satisfy the rank conditions of the Implicit Function Theorem, suggesting that its required injectivity condition is not satisfied. The (possibly negative) results of this research will be reported in the near future.

\bibliographystyle{IEEEtran}

\end{document}